\documentclass[journal,onecolumn]{IEEEtran}
\usepackage{longtable}
\usepackage{multirow}
\usepackage{tabularx}
\usepackage{lipsum}
\usepackage{multicol}
\usepackage{amssymb}
\usepackage{amsmath}
\usepackage{longtable}
\newtheorem{theorem}{Theorem}[section]
\newtheorem{lemma}[theorem]{Lemma}
\newtheorem{proposition}[theorem]{Proposition}
\newtheorem{corollary}[theorem]{Corollary}

\usepackage{xcolor}

\newenvironment{proof}[1][Proof]{\begin{trivlist}
\item[\hskip \labelsep {\bfseries #1}]}{\end{trivlist}}

\makeatletter

\newcommand{\Rmnum}[1]{\expandafter\@slowromancap\romannumeral #1@}

\makeatother

\ifCLASSINFOpdf
\else
\fi
\begin{document}

%
\title{Linear codes over $\mathbb F_q$ which are equivalent to  LCD codes}

\author{ Claude Carlet$^1$ \and Sihem Mesnager$^2$ \and Chunming Tang$^3$ \and  Yanfeng Qi$^4$
\thanks{This work was supported by SECODE project and
the National Natural Science Foundation of China
(Grant No. 11401480, 11531002). C. Tang
also acknowledges support from 14E013 and
CXTD2014-4 of China West Normal University.
Y. Qi also acknowledges support from Zhejiang provincial Natural Science Foundation of China (LQ17A010008).
}

\thanks{C. Carlet is with Department of Mathematics, Universities of Paris VIII and XIII, LAGA, UMR 7539, CNRS, Sorbonne Paris Cit\'{e}. e-mail: claude.carlet@univ-paris8.fr}

\thanks{S. Mesnager is with Department of Mathematics, Universities of Paris VIII and XIII and Telecom ParisTech, LAGA, UMR 7539, CNRS, Sorbonne Paris Cit\'{e}. e-mail: smesnager@univ-paris8.fr}
\thanks{C. Tang is with School of Mathematics and Information, China West Normal University, Nanchong, Sichuan,  637002, China. e-mail: tangchunmingmath@163.com
}

\thanks{Y. Qi is with School of Science, Hangzhou Dianzi University, Hangzhou, Zhejiang, 310018, China.
e-mail: qiyanfeng07@163.com
}

}

%


\maketitle

\begin{abstract}
Linear codes with complementary duals (abbreviated LCD) are linear codes whose intersection with their dual are trivial.
When they are binary, they play an important role in armoring implementations against side-channel attacks and fault injection attacks. Non-binary LCD codes in characteristic 2 can be transformed into binary LCD codes by expansion. In this paper, we introduce a general construction of LCD codes from any linear codes. Further, we show that any linear code over
$\mathbb F_{q} (q>3)$ is equivalent to an Euclidean LCD code and any linear code over
$\mathbb F_{q^2} (q>2)$ is equivalent to a Hermitian LCD code. Consequently an $[n,k,d]$-linear Euclidean LCD code over $\mathbb F_q$ with $q>3$ exists if there is an $[n,k,d]$-linear code over $\mathbb F_q$ and an $[n,k,d]$-linear Hermitian LCD code over $\mathbb F_{q^2}$ with $q>2$ exists if there is an $[n,k,d]$-linear code over $\mathbb F_{q^2}$. Hence, when $q>3$ (resp.$q>2$)  $q$-ary Euclidean
(resp. $q^2$-ary Hermitian) LCD codes possess the same asymptotical bound  as $q$-ary linear codes (resp. $q^2$-ary linear codes). Finally,
we present an approach of constructing LCD codes by extending linear codes.

\end{abstract}

\begin{IEEEkeywords}
 Linear codes,  complementary dual, LCD codes,
 Euclidean LCD codes, Hermitian LCD codes
\end{IEEEkeywords}

%
\IEEEpeerreviewmaketitle

\section{Introduction}

A linear  complementary dual  code (abbreviated LCD) is defined as a linear code $\mathcal C$ whose dual code $\mathcal C ^ \perp$ satisfies $\mathcal C \cap \mathcal C^ \perp=\{0\}$.
LCD codes have been widely applied in data storage, communications systems, consumer electronics, and cryptography. In \cite{mas92}, Massey showed that LCD codes  provide an optimum linear coding solution for the two-user binary adder channel. Recently, Carlet and Guilley  \cite{CG14} investigated an interesting application of binary LCD codes against side-channel attacks (SCA) and fault injection attacks (FIA) and presented several constructions of LCD codes. They showed in particular that non-binary LCD codes in characteristic 2 can be transformed into binary LCD codes by expansion.  It is then important to keep in mind that, for SCA, the most interesting case is when $q$ is even.

LCD codes are also interesting objects in the general framework of algebraic coding. For asymptotical optimality and bounds of LCD codes, Massey \cite{mas92} showed that there exist asymptotically good LCD codes. Tzeng and Hartmann \cite{TH70} proved that the minimum distance of a class of LCD codes is greater than that given by the BCH bound. Sendrier \cite{Sen04} showed that LCD codes meet the asymptotic Gilbert-Varshamov bound using properties of the hull dimension spectrum of linear codes. Dougherty et al. \cite{DKO15} gave a linear programming bound on the largest size of an LCD code of given length and minimum distance. Recently, Galvez et al. \cite{GKL17} studied the maximum minimum distance of LCD codes of fixed length and dimension.

Many works have been devoted to the characterization and constructions of LCD codes.  Yang and Massey  provided in \cite{YM94} a necessary and sufficient condition under which a cyclic code has a complementary dual. In \cite{GOS16}, quasi-cyclic codes that are LCD have been characterized and studied using their concatenated structures. Criteria for complementary duality of  generalized quasi-cyclic codes (GQC) bearing on the component codes are given and some explicit long GQC that are LCD, but not quasi-cyclic, have been exhibited in \cite{OOS17}.  In \cite{DNS17}, Dinh, Nguyend and Sriboonchitta  investigated the algebraic structure of $\lambda$-constacyclic codes over  finite commutative semi-simple rings. Among others, necessary and sufficient conditions for the existence of  LCD,   $\lambda$-constacyclic codes over such finite semi-simple rings have been provided. In \cite{DLL16}, Ding et al. constructed several families of LCD cyclic codes over finite fields and analyzed their parameters. In \cite{LDL16_1} Li et al. studied a class of LCD BCH codes proposed in \cite{LDL16_0} and extended the results on their parameters.
Mesnager et al. \cite{MTQ16} provided a construction of algebraic geometry LCD codes which could be good candidates to be resistant against SCA.  Liu and Liu constructed LCD matrix-product codes using quasi-orthogonal matrices in \cite{LL16}. It was also shown by Kandasamy et al. \cite{KSS12} that maximum rank distance codes generated by the trace-orthogonal-generator matrices are LCD codes. However, little is known on Hermitian LCD codes. More precisely,  it has been  proved in \cite{GOS16} that those codes are asymptotically good. By employing their generator matrices, Boonniyoma and Jitman gave in \cite{BJ16} a sufficient and necessary condition on Hermitian codes for being LCD. Li \cite{Li17} constructed some cyclic Hermitian LCD codes over finite fields and analyzed their parameters.


 In \cite{Jin16}, Jin showed that some Reed-Solomon codes are equivalent to LCD codes. In \cite{CMTQ17}, the
authors proved that any MDS code is  equivalent to an LCD code. Very recently, Jin and Xing \cite{JX17} showed that an algebraic geometry code over $\mathbb F_{2^m} (m\ge 7)$ is equivalent to an LCD code. As a
consequence, they proved that there exists a family of LCD codes which is equivalent to algebraic geometry codes exceeding the asymptotical Gilbert-Varshamov bound.

The main goal of this manuscript  is to study all possible Euclidean and Hermitian LCD  codes. We completely determine all Euclidean LCD codes over $\mathbb F_q (q>3)$ and all Hermitian LCD codes over $\mathbb F_{q^2} (q>2)$   for all possible parameters.
More precisely,  we introduce a general construction of LCD codes from any linear codes. Further, we show that any linear code over
$\mathbb F_{q} (q>3)$ is equivalent to an Euclidean LCD code and any linear code over
$\mathbb F_{q^2} (q>2)$ is equivalent to a Hermitian LCD code. Consequently an $[n,k,d]$-linear   Euclidean LCD code over $\mathbb F_q$ with $q>3$ exists if there is an $[n,k,d]$-linear code over $\mathbb F_q$ and an $[n,k,d]$-linear Hermitian LCD code over $\mathbb F_{q^2}$ with $q>2$ exists if there is an $[n,k,d]$-linear code over $\mathbb F_{q^2}$. Hence, when $q>3$, $q$-ary Euclidean LCD codes are as good as $q$-ary linear codes and when $q>2$, $q^2$-ary Hermitian
LCD codes are as good as $q^2$-ary linear codes.

The paper is organized as follows.
Section \ref{Pre} gives preliminaries and background on  Euclidean and Hermitian LCD codes.
In Section \ref{det-matrix}, we present a result about matrices.  In  Section \ref{LCD-code}, we firstly provide a construction of Euclidean (resp. Hermitian) LCD codes from any linear codes.
 Based on these results  we completely determine all Euclidean and Hermitian LCD MDS codes for all possible parameters. In addition, we also show one can construct LCD codes by
 extending linear codes.
\section{Preliminaries}\label{Pre}

Throughout this paper,  $p$ is a prime and $\mathbb F_q$  is the finite field of order $q$, where $q=p^m$ for some positive integer $m$. The set of non-zero elements of $\mathbb F_q$ is denoted by $\mathbb F^\times_q$. For any $x \in \mathbb F_{q^2}$, the conjugate of $x$ is defined as $\overline{x}= x^{q}$. For a matrix $A$,
$A^{T}$  denotes the transposed matrix of matrix $A$, $\overline A$  denotes the conjugate matrix of $A$ and $\text{Rank}(A)$ denotes the rank of $A$.  When $A$ is a square matrix, $\det A$ denotes the determinant of
$A$. We denote by $\# I$  the cardinality of a finite set $I$.
An $[n, k, d]$ linear code $\mathcal C$ over $\mathbb F_q$ is a linear subspace of
$\mathbb F_q$ with dimension $k$ and minimum (Hamming) distance $d$. The value $n-k$ is called the \emph{codimension} of $\mathcal C$.
Given a linear code $\mathcal C$ of length $n$ over $\mathbb F_{q}$ (resp.  $\mathbb F_{q^2}$), its Euclidean dual code (resp. Hermitian dual code) is denoted by $\mathcal C^ {\perp}$ (resp. $\mathcal C^ {\perp_H}$). The codes $\mathcal C^ {\perp}$  and $\mathcal C^ {\perp_H}$ are defined by
\[\mathcal C^ {\perp}=\{ (b_0, b_1, \cdots, b_{n-1})\in \mathbb F_{q}^n:  \sum _{i=0}^{n-1} b_i  c_i=0, \forall (c_0, c_1, \cdots, c_{n-1}) \in \mathcal C \}, \]

\[\mathcal C^ {\perp_H}=\{ (b_0, b_1, \cdots, b_{n-1})\in \mathbb F_{q^2}^n:  \sum _{i=0}^{n-1} b_i \overline c_i=0, \forall (c_0, c_1, \cdots, c_{n-1}) \in \mathcal C \}, \]
respectively.

 The minimum distance of an $[n, k, d]$ linear code is bounded by the Singleton bound
\[d\le n+1-k.\]
A  code meeting the above bound is called \emph{Maximum Distance Separable} (MDS).

The Euclidean (resp. Hermitian) hull of a linear code $\mathcal C$ is defined to be $\text{Hull}_{E}(\mathcal C):= \mathcal C \cap \mathcal C^ {\perp}$
(resp. $\text{Hull}_{H}(\mathcal C):=\mathcal C \cap \mathcal C^ {\perp_H}$). Let $h_{E}(\mathcal C)$ (resp. $h_{H}(\mathcal C)$) be the dimension of $\text{Hull}_{E}(\mathcal C)$
(resp. $\text{Hull}_{H}(\mathcal C)$).
A linear code $\mathcal C$ over $\mathbb F_{q}$ is called an \emph{LCD code} (or for short, LCD code) if $h_{E}(\mathcal C)=0$.
A linear code $\mathcal C$ over $\mathbb F_{q^2}$ is called a \emph{Hermitian LCD code} (linear code with Hermitian complementary dual) if  $h_{H}(\mathcal C)=0$. To distinguish between cassical LCD codes and Hermitian ones, we shall call  \emph{Euclidean LCD codes} in the former case.
The following  proposition gives a complete characterization of
 Euclidean and Hermitian LCD codes (see. \cite{BJ16,CG14}).
\begin{proposition}\label{LDC}
 If $G$ is a generator matrix for the $[n, k]$ linear code $\mathcal C$, then $\mathcal C$ is an Euclidean (resp. a Hermitian) LCD code if and only if, the $k \times k$ matrix $G G^{T}$
 (resp. $G\overline{G}^{T}$) is nonsingular.
\end{proposition}

For any prime power $q$, let $\alpha_q^{lin}$ and $\alpha_q^{E}$ denote the functions which are defined by
\begin{align*}
\alpha_q^{lin}(\delta):=\sup \{R\in [0,1] : (\delta, R)\in U_q^{lin}\}, \text{for } 0\le \delta \le 1,
\end{align*}
and
\begin{align*}
\alpha_q^{E}(\delta):=\sup \{R\in [0,1] : (\delta, R)\in U_q^{E}\}, \text{for } 0\le \delta \le 1.
\end{align*}
Here $U_q^{lin}$ (resp. $U_q^{E}$) is the set of all ordered pairs $(\delta, R)\in [0,1]^2$ for which there exists a sequence of $[n_i,k_i, d_i]$ linear code (resp. Euclidean LCD code) $\mathcal C_i$ over $\mathbb F_q$ such that $n_i \rightarrow \infty$ as $i \rightarrow \infty$ and
\begin{align*}
\delta=\lim_{i \rightarrow \infty}  \frac{d_i}{n_i},~~R=\lim_{i \rightarrow \infty}  \frac{k_i}{n_i}.
\end{align*}
$\alpha_q^{lin}$  (resp. $\alpha_q^{E}$) is the largest asymptotic information rate that can be achieved for a given asymptotic relative minimum distance $\delta$ of $q$-ary linear codes
(resp. Euclidean LCD codes). We can define $\alpha_{q^2}^{H}(\delta)$ for Hermitian LCD codes over $\mathbb F_{q^2}$ in a similar way. It is trivial to see
that $\alpha_{q}^{lin}(\delta)\ge \alpha_{q}^{E}(\delta)$ and $\alpha_{q^2}^{lin}(\delta)\ge \alpha_{q^2}^{H}(\delta)$.

A classical lower bound for $\alpha_{q}^{lin}$ is the asymptotical Gilbert-Varshamov (GV) bound \cite{MS77}.
\begin{proposition}
For any prime power $q$, we have
\begin{align*}
\alpha_q^{lin} \ge R_{GV}(\delta)=1-H_q(\delta), \text{ for } 0 < \delta < \frac{q-1}{q},
\end{align*}
where $H_q(x)=x\log_{q}(q-1)-x \log_{q}(x)-(1-x) \log_{q}(1-x)$ is the $q$-ary entropy function.
\end{proposition}
The following TVZ bound, which is better than Varshamov-Gilbert bound, was established in \cite{TVZ82} by using algebraic-geometry codes.
\begin{proposition}\label{TVZ}
Let $q$ be a prime power. Then,
\begin{align*}
\alpha_q^{lin}(\delta)\ge 1-\delta-\frac{1}{A(q)}, \text{ for } \delta\in [0,1],
\end{align*}
where $A(q) = \limsup_{g\rightarrow \infty} \frac{N_q(g)}{g}$ and $N_q(g)$ denotes the maximum number of rational places that a global function field
of genus $g$ with full constant field $\mathbb F_q$ can have.
\end{proposition}

It was proved that LCD codes can also attain the  asymptotical GV bound in \cite{Sen04}, i.e. $\alpha_q^{E}(\delta)\ge R_{GV}(\delta)$.

Very recently, Jin and Xing \cite{JX17} proved that $\alpha_q^{E}(\delta)$ has bound better than the GV bound in some special cases by using algebraic-geometry codes. They proved the following result.
\begin{proposition}
When $q\ge 128$ is a power of $2$, then $\alpha_q^{E}(\delta)$
exceeds the asymptotic Gilbert-Varshamov bound in two intervals of $[0,1]$.
\end{proposition}

In this paper, we will prove that for any prime power $q$ with $q>3$ and $\delta \in [0,1]$, $\alpha_q^{E}(\delta)=\alpha_{q}^{lin}(\delta)$ holds. Hence, any lower bound of
$\alpha_{q}^{lin}(\delta)$ is also a lower bound of $\alpha_{q}^{E}(\delta)$.

For any $\mathbf a=(a_1,\cdots, a_n) \in \mathbb F_q^{n}$ and permutation $\sigma$ of $\{1,2,\cdots, n\}$, we define $\mathcal C_{\mathbf a}$ and $\sigma(\mathcal C)$ as the following linear codes
\begin{align*}
\mathcal C_{\mathbf a}=\{(a_1c_1,\cdots, a_n c_n): (c_1, \cdots, c_n) \in \mathcal C\},
\end{align*}
and
\begin{align*}
\sigma(\mathcal C)=\{(c_{\sigma(1)},\cdots, c_{\sigma(n)}): (c_1, \cdots, c_n) \in \mathcal C\}.
\end{align*}

Two codes $\mathcal C$ and $\mathcal C'$ in $\mathbb F_q^n$ are called \emph{equivalent} if $\mathcal C'=\sigma(\mathcal C_{\mathbf a})$ for some permutation $\sigma$ of $\{1,2,\cdots, n\}$ and $\mathbf a \in \mathbb (\mathbb F_q^*)^n$.  Any $[n,k]$ linear code over a field is equivalent to a code generated by a matrix of the form $[I_k: P]$ where $I_k$ denotes the $k\times k$ identity matrix.

\section{Some results on matrices}\label{det-matrix}

Let $l$ and  $j$ be two integers with $0\le j\le l$ and $l\ge 1$. Let $M$ be an $l\times l$ matrix over $\mathbb F_q$. Let ${\bf u}$ be a word in $\mathbb F_q^l$ of Hamming weight $j$ and $I=\{i_1,\cdots, i_j\}$ its support. Denote by $diag_{\, l} [{\bf u}]$ the diagonal $l\times l$ matrix whose elements on the diagonal are $u_1,\dots ,u_l$. Define $M_{I}$ the submatrix of $M$ obtained by deleting the $i_1,\cdots, i_j$-th rows and columns of $M$. Denote $M_{I}=1$ if $I=\{1,2, \cdots, l\}$ and $M_{\emptyset}=M$.

\begin{lemma}\label{matrix}
Let $M$ be an  $l\times l$ matrix over $\mathbb F_q$ and $t$ an integer with $0\le t\le l-1$.
Suppose that $\det(M_I)=0$ holds for any  subset  $I$ of $\{1,2,\cdots, l\}$ with $0\le \# I \le t$.
Then,  for any $1\le j \le t+1$ and every word ${\bf u}$ of Hamming weight $j$, denoting its support by $J$, we have:
\begin{align}\label{MM}
\det(M+diag_{\, l} [{\bf u}])=
{\Big (}\prod_{i\in J} u_i{\Big )} \det(M_J).
\end{align}
\end{lemma}
\begin{proof}
We will prove this statement by induction on $t$. We first prove that the statement holds if $t=0$. In this case, $\det (M)=0$. Then, for any ${\bf u}$ of Hamming weight 1, denoting by $i_1$ the position of its only nonzero coordinate, we have:
\begin{align*}
\det(M+diag_{\, l} [{\bf u}])=\det(M)+u_{i_1} \det(M_{\{i_1\}})=u_{i_1} \det(M_{\{i_1\}}).
\end{align*}
Thus, the statement holds if $t=0$. In the following, suppose the statement holds for $t=0,1,2,\cdots, s$. We will prove that the statement holds for $t=s+1$.
We only need to check the property for a word ${\bf u}$ of Hamming weight $s+2$. Let the support of $\mathbf u$ be $J=\{i_1,\cdots, i_{s+2}\}$. Note that if ${\bf u}'$ is the word obtained from ${\bf u}$ by changing $u_{i_{s+2}}$ into 0 and $\mathbf u''$ is the word by deleting the $i_{s+2}$ component of $\bf u$, we have:
\begin{align*}
\det(M+diag_{\, l} [{\bf u}'])=0
\end{align*}
and
\begin{align*}
\det(M+diag_{\, l} [{\bf u}])=&
\det(M+diag_{\, l} [{\bf u}'])+u_{i_{s+2}} \det(M_{\{i_{s+2}\}}+diag_{\, l-1} [{\bf u}''])\\
=&u_{i_{s+2}} \det(M_{\{i_{s+2}\}}+diag_{\, l-1} [{\bf u}''])\\
=& u_{i_1} u_{i_2} \cdots u_{i_{s+2}} \det(M_{\{i_1,  i_2, \cdots ,  i_{s+2}\}}).
\end{align*}
Thus, the Equation (\ref{MM}) holds for $j=s+2$. When $1\le j\le s+1$, the Equation (\ref{MM}) holds from the inductive assumption.\\
This completes the proof.
\end{proof}

For any matrix $M$, $\text{Row}(M)$ (resp. $\text{Col}(M)$) denotes the vector space spanned by the rows (resp. columns) of $M$. Then,
$\text{dim}(\text{Row}(M))= \text{dim}(\text{Col}(M))$ and denote it by \text{Rank}(M), which is called the rank of $M$. For $\mathbf v^i\in \mathbb F^k$ ($i \in \{1,2, \cdots, n\}$),
$\text{Span}\{\mathbf v^1, \cdots, \mathbf v^n\}$ denotes the vector space spanned by $\mathbf v^1, \cdots, \mathbf v^n$.


\begin{lemma}\label{rank-GG}
Given a matrix $G$ with $k$ rows and $n$ columns, one has $\text{Rank}(GG^T)\leq \text{Rank}(G)$.
\end{lemma}
\begin{proof}
Let $\mathbf{g}^1, \cdots, \mathbf{g}^n$ be the columns vectors of $G$. Then $GG^{T}=\sum_{i=1}^{n} \mathbf g^i (\mathbf g^{i})^{T}$. By
$\text{Col}(\mathbf g^i (\mathbf g^{i})^{T}) =\text{Span}\{\mathbf g^i\}$, $\text{Col}(GG^{T}) \subseteq \text{Span}\{\mathbf g^i : i\in \{1, \cdots, n\}\}$. This corollary follows
from  the rank of matrix.
\end{proof}

\section{Construction of LCD codes from any linear code}\label{LCD-code}
\subsection{Construction of LCD codes which are equivalent to the original one}
Let $\mathcal C$ be an $[n,k,d]$-linear code over $\mathbb F_q$ with generator matrix $G=[I_k: P]$. For any $\mathbf a=(a_1,\cdots, a_n) \in \mathbb F_q^{n}$, we define $\mathcal C_{\mathbf a}$ as the following linear code
\begin{align*}
\mathcal C_{\mathbf a}=\{(a_1c_1,\cdots, a_n c_n): (c_1, \cdots, c_n) \in \mathcal C\}.
\end{align*}


\begin{theorem}\label{LCD-Euclidean}
Let $q$ be a power of a prime  and let $\mathcal C$ be an $[n,k,d]$-linear code over $\mathbb F_q$ with generator matrix $G=[I_k: P]$. Let $M=GG^{T}$.
Let $t\leq k-1$ be a non-negative integer such that $\det(M_I)=0$ for any subset of $\{1,2, \cdots, k\}$ with $0 \le \# I \le t$ and suppose there exists $J\subseteq \{1,\dots ,k\}$ of size $t+1$ such that
$\det(M_J)\neq 0$. Let ${\bf a}$ be any word of length $n$ such that $a_j\in \mathbb F_q\setminus \{1, -1\}$ for $j\in J$ and $a_j\in \{1, -1\}$ for $j\in \{1,\dots ,n\}\setminus J$, then $\mathcal C_{\mathbf a}$
is an Euclidean complementary dual  $[n,k]$-linear code. Furthermore, if
$a_j\neq 0$,  then $\mathcal C_{\mathbf a}$
is an Euclidean complementary dual  $[n,k,d]$-linear code.
\end{theorem}
\begin{proof}
Let $G'$ be the generator matrix of $\mathcal C_{\mathbf a}$ obtained from $G$ by multiplying its $j$-th column by $a_j$ for $j\in \{1,2,\cdots n\}$ and let ${\bf u}$ be the word of length $n$ and support $J$ such that $u_j=a_j^2-1$ for $1\le j \le k$.
Then, $G'G'^T=M+diag_{\, k}[{\bf u}]$. From Lemma \ref{matrix},
\begin{align*}
\det(G'G'^T)=&\det(M+diag_{\, k}[{\bf u}])\\
=& {\Big (}\prod_{j\in J}u_j{\Big )}\det(M_J) \neq 0.\\
\end{align*}
By Proposition \ref{LDC}, $\mathcal C_{\mathbf a}$ is an Euclidean LCD code.
\end{proof}

\begin{proposition}
Let $q$ be a power of a prime  and let $\mathcal C$ be an $[n,k,d]$-linear code over $\mathbb F_q$ with generator matrix $G=[I_k: P]$. Let $M=GG^{T}$.
Let $t\leq k-1$ be a non-negative integer such that $\det(M_I)=0$ for any subset $I$ of $\{1,2, \cdots, k\}$ with $0 \le \# I \le t$ and suppose there exists $J\subseteq \{1,\dots ,k\}$ of size $t+1$ such that
$\det(M_J)\neq 0$. Then, $t\ge h_{E}(\mathcal C)-1$.
\end{proposition}
\begin{proof}
Let $J\subseteq \{1,\dots ,k\}$ be any size $t+1$ such that
$\det(M_J)\neq 0$. Let $\mathcal C_1:=\text{Span}\{\mathbf b_j: j\in \{1,2,\cdots, n\} \setminus J\}$, where $\mathbf b_j$ is the $j$-row of $G$ and $G_1$ is obtained from $G$
by deleting the $j$-row of $G$ for $j \in J$. Then, $G_1$ is a generator matrix
of $\mathcal C_1$ and $\det (G_1G_1^{T})=\det (M_J)\neq 0$. Thus, $\mathcal C_1$ is an LCD code.\\
Suppose that $1\le t+1< h_{E}(\mathcal C)$. Then, $\text{dim}(\mathcal C_1)+\text{dim}(\text{Hull}_{E}(\mathcal C))=k-(t+1)+ h_{E}(\mathcal C)\ge k+1>\text{dim}(\mathcal C)$.
From $\mathcal C_1 \subseteq  \mathcal C$ and $\text{Hull}_{E}(\mathcal C) \subseteq \mathcal C$, $\mathcal C_1 \cap \text{Hull}_{E}(\mathcal C) \neq \{0\}$.
Choose any $\mathbf c \in \mathcal C_1 \cap \text{Hull}_{E}(\mathcal C) $. It's easy to observe that $\mathbf c \in \mathcal C_1 \cap \mathcal C_1^{\perp}$, which  contradicts with
$\mathcal C$ being LCD.
\end{proof}

The reader notices that the above theorem generalizes Theorem 3.3 of \cite{CMTQ17}.
\begin{corollary}
Let $q$ be a power of a prime with $q>3$ and $\mathcal C$ be an $[n,k,d]$-linear code over $\mathbb F_q$. Then, there exists $\mathbf a \in (a_1,\cdots, a_n) \in \mathbb F_q^n$ with
$a_j \neq 0$ for any $1 \le j \le n$ such that $\mathcal C_{\mathbf a}$ is an Euclidean LCD code.
\end{corollary}
\begin{proof}
If $\mathcal C$ is an Euclidean LCD code, the result holds by choosing $\mathbf a=(1,1,\cdots,1)$.\\
If $\mathcal C$ is not an Euclidean LCD code, then $\det (GG^{T})=0$. Let $M=GG^{T}$. Then there exists a non-negative integer $t$ and a
subset $J$ of $\{1,2,\cdots,k\}$ with $\# J=t+1$ such that
$\det(M_I)=0$ for any subset of $\{1,2, \cdots, k\}$ with $0 \le \# I \le t$ and $\det(M_J)\neq 0$.
Since $q>3$, $\mathbb F_q^*\setminus \{1,-1\}\neq \emptyset$. Thus, one can choose $\mathbf a\in \mathbb F_q^n$ with $a_j\in \mathbb F_q^*\setminus \{-1, 1\}$ for $j\in J$ and $a_j=1$ for $j\in \{1,\dots ,n\}\setminus J$. From Theorem \ref{LCD-Euclidean}, $\mathcal C_{\bf a}$ is an Euclidean LCD code.
\end{proof}

\begin{corollary}
Let $q$ be a power of a prime with $q>3$. Then, an $[n,k,d]$-linear Euclidean LCD code over $\mathbb F_q$ exists if there is an $[n,k,d]$-linear code over $\mathbb F_q$.
\end{corollary}

\begin{corollary}
Let $q$ be a power of a prime with $q>3$. Then, $\alpha_q^{E}(\delta)=\alpha_q^{lin}(\delta)$ for any $\delta \in [0,1]$. In particular,
\begin{align*}
\alpha_q^{E}(\delta)\ge 1-\delta-\frac{1}{A(q)}, \text{ for } \delta\in [0,1],
\end{align*}
where $A(q) $ is defined as in Proposition \ref{TVZ}.
\end{corollary}

\begin{theorem}\label{LCD-Hermitian}
Let $q$ be a power of a prime, let $\mathcal C$ be an $[n,k,d]$-linear code over $\mathbb F_{q^2}$ with generator matrix $G=[I_k: P]$ and $M=GG^{T}$.
Suppose $t\leq k-1$ is a non-negative integer such that $\det(M_I)=0$ for any subset $I$ of $\{1,2, \cdots, k\}$ with $0 \le \# I \le t$ and there exists $J\subseteq \{1,\dots ,k\}$ of size $t+1$ such that
$\det(M_J)\neq 0$. Let ${\bf a}$ be any word of length $n$ such that $a_j\in \mathbb F_{q^2} \backslash  \mathbb (\mathbb F_{q^2}^*)^{q-1}$ for $j\in J$ and $a_j\in (\mathbb F_{q^2}^*)^{q-1}$ for $j\in \{1,\dots ,n\}\setminus J$, then $\mathcal C_{\mathbf a}$
is an $[n,k]$-linear code with Hermitian complementary dual. Furthermore, if
$a_j\neq 0$, then $\mathcal C_{\mathbf a}$
is an $[n,k,d]$-linear code with Hermitian complementary dual.
\end{theorem}
\begin{proof}
Let $G'$ be the generator matrix of $\mathcal C_{\mathbf a}$ obtained from $G$ by multiplying its $j$-th column by $a_j$ for $j\in \{1,2,\cdots n\}$ and let ${\bf u}$ be the word of length $n$ and support $J$ such that $u_j=a_j^{q+1}-1$ for $1\le j \le k$.
Then, $G'\overline{G'}^T=M+diag_{\, k}[{\bf u}]$. From Lemma \ref{matrix},
\begin{align*}
\det(G'\overline{G'}^T)=&\det(M+diag_{\, k}[{\bf u}])\\
=& {\Big (}\prod_{j\in J}u_j{\Big )}\det(M_J)\neq 0.\\
\end{align*}
By Proposition \ref{LDC}, $\mathcal C_{\mathbf a}$ is a Hermitian LCD code.
\end{proof}

\begin{proposition}
Let $q$ be a power of a prime  and let $\mathcal C$ be an $[n,k,d]$-linear code over $\mathbb F_q$ with generator matrix $G=[I_k: P]$. Let $M=G\overline{G}^{T}$.
Let $t\leq k-1$ be a non-negative integer such that $\det(M_I)=0$ for any subset of $\{1,2, \cdots, k\}$ with $0 \le \# I \le t$ and suppose there exists $J\subseteq \{1,\dots ,k\}$ of size $t+1$ such that
$\det(M_J)\neq 0$. Then, $t\ge h_{H}(\mathcal C)-1$.
\end{proposition}
\begin{proof}
Let $J\subseteq \{1,\dots ,k\}$ be any size $t+1$ such that
$\det(M_J)\neq 0$. Let $\mathcal C_1:=\text{Span}\{\mathbf b_j: j\in \{1,2,\cdots, n\} \setminus J\}$, where $\mathbf b_j$ is the $j$-row of $G$ and $G_1$ is obtained from $G$
by deleting the $j$-row of $G$ for $j \in J$. Then, $G_1$ is a generator matrix
of $\mathcal C_1$ and $\det (G_1\overline{G}_1^{T})=\det (M_J)\neq 0$. Thus, $\mathcal C_1$ is a Hermitian LCD code.\\
Suppose that $1\le t+1< h_{H}(\mathcal C)$. Then, $\text{dim}(\mathcal C_1)+\text{dim}(\text{Hull}_{H}(\mathcal C))=k-(t+1)+ h_{H}(\mathcal C)\ge k+1>\text{dim}(\mathcal C)$.
From $\mathcal C_1 \subseteq  \mathcal C$ and $\text{Hull}_{H}(\mathcal C) \subseteq \mathcal C$, $\mathcal C_1 \cap \text{Hull}_{H}(\mathcal C) \neq \{0\}$.
Choose any $\mathbf c \in \mathcal C_1 \cap \text{Hull}_{H}(\mathcal C) $. It's easy to observe that $\mathbf c \in \mathcal C_1 \cap \mathcal C_1^{\perp_H}$, which  contradicts with
$\mathcal C$ being Hermitian LCD.
\end{proof}

\begin{corollary}
Let $q$ be a power of a prime with $q>2$ and $\mathcal C$ be an $[n,k,d]$-linear code over $\mathbb F_{q^2}$. Then, there exists $\mathbf a \in (a_1,\cdots, a_n) \in \mathbb F_{q^2}^n$ with
$a_j \neq 0$ for any $1 \le j \le n$ such that $\mathcal C_{\mathbf a}$ is a Hermitian LCD code.
\end{corollary}
\begin{proof}
If $\mathcal C$ is a Hermitian LCD code, the result holds by choosing $\mathbf a=(1,1,\cdots,1)$.\\
If $\mathcal C$ is not a Hermitian LCD code, then $\det (G\overline{G}^{T})=0$. Let $M=G\overline{G}^{T}$. Then there exists a non-negative integer $t$ and a
subset $J$ of $\{1,2,\cdots,k\}$ with $\# J=t+1$ such that
$\det(M_I)=0$ for any subset of $\{1,2, \cdots, k\}$ with $0 \le \# I \le t$ and $\det(M_J)\neq 0$.
Since $q>2$, $\mathbb F_{q^2}\setminus \mathbb F_{q^2}^{q-1}\neq \emptyset$. Thus, one can choose $\mathbf a\in \mathbb F_{q^2}^n$ with $a_j\in \mathbb F_{q^2}\setminus \mathbb F_{q^2}^{q-1}$ for $j\in J$ and $a_j=1$ for $j\in \{1,\dots ,n\}\setminus J$. From Theorem \ref{LCD-Hermitian}, $\mathcal C_{\bf a}$ is a Hermitian LCD code.
\end{proof}

\begin{corollary}
Let $q$ be a power of a prime with $q>2$. Then, an $[n,k,d]$-linear Hermitian LCD code over $\mathbb F_{q^2}$ exists if there is an $[n,k,d]$-linear code over $\mathbb F_{q^2}$.
\end{corollary}

\begin{corollary}
Let $q$ be a power of a prime with $q>2$. Then, $\alpha_{q^2}^{H}(\delta)=\alpha_{q^2}^{lin}(\delta)$ for any $\delta \in [0,1]$. In particular,
\begin{align*}
\alpha_{q^2}^{H}(\delta)\ge 1-\delta-\frac{1}{A(q^2)}, \text{ for } \delta\in [0,1],
\end{align*}
where $A(q^2) $ is defined as in  Proposition \ref{TVZ}.
\end{corollary}

\subsection{Construction of LCD codes by extending linear codes}

\begin{lemma}
Let $\mathcal C$ be an $[n,k]$-linear over $\mathbb F_{q}$ with $h:=h_{E}(\mathcal C)>0$. Let $\{\mathbf c^i\}_{i=1}^{k}$ be a basis of $\mathcal C$ such that $\{\mathbf c^i\}_{i=1}^{h}$ is a basis
 of $\text{Hull}_{E}(\mathcal C)$. Then $\text{Span}\{\mathbf c^i: i=h+1,h+2, \cdots, k\}$ is an Euclidean LCD code.
\end{lemma}
\begin{proof}
Let $\mathcal C_1 =\text{Span}\{\mathbf c^i: i=h+1,h+2, \cdots, k\}$ and $\mathbf c \in \mathcal C_1 \cap \mathcal C_1^{\perp}$. Then, for any
$\mathbf c'=(c'_1, \cdots, c'_n) \in \mathcal C_1$ and $\mathbf c''=(c''_1, \cdots, c''_n) \in \text{Hull}_{E}(\mathcal C)$, $\sum_{i=1}^n c_ic'_i=0$ and $\sum_{i=1}^n c_ic''_i=0$.
Thus, $\mathbf c \in \mathcal C^{\perp}$. Hence, $\mathbf c =0$ from $\mathcal C_1 \cap \mathcal C^{\perp}=\{0\}$. It completes the proof of this lemma.

\end{proof}
Let $\mathcal C$ be a linear code and $t$ be a positive integer. $L=(l_1,\cdots, l_t)$ denotes a linear transform from $\mathcal C$ to $\mathbb F_q^t$ defined by
\begin{align*}
\mathbf c \longmapsto (l_1(\mathbf c), \cdots, l_t(\mathbf c)), \text{ for } \mathbf c\in \mathcal C,
\end{align*}
where $l_i$ is any linear form over $\mathcal C$  for
 $i \in \{1,2, \cdots, t\}$.
Define
\begin{align}\label{C-L}
\mathcal C_{L}:=\{(\mathbf c, l_1(\mathbf c), \cdots, l_t(\mathbf c)): \mathbf c\in \mathcal C\}.
\end{align}
Let $\text{Ker} (L):=\{\mathbf c\in \mathcal C: L(\mathbf c)=0\}$, $\text{Im}(L):= \{L(\mathbf c): \mathbf c \in \mathcal C\}$ and $\text{Rank}(L):=\text{dim}(\text{Im}(L) )$.\\
In the following statement, we give a sufficient condition such that  the code $\mathcal C_L$ does not give rise to an  Euclidean LCD code.
\begin{theorem}
Let $\mathcal C$ be an $[n,k]$-linear over $\mathbb F_{q}$ with $h:=h_{E}(\mathcal C)>0$ and $L$ be a linear transform from $\mathcal C$ to $\mathbb F_q^t$ such that
$\mathcal C= \text{Ker} (L) + \text{Hull}(\mathcal C)$. If $\text{Rank}(L)<h$, then $\mathcal C_L$ (defined by Equation (\ref{C-L})) is not an Euclidean LCD code.\end{theorem}
\begin{proof}
From $\mathcal C= \text{Ker} (L) + \text{Hull}(\mathcal C)$, there is   a basis $\{\mathbf c^i\}_{i=1}^{k}$ of $\mathcal  C$ such that
 $\{\mathbf c^i\}_{i=1}^{h}$ is a basis of $\text{Hull}_{E}(\mathcal C)$ and $\{\mathbf c^i\}_{i=h+1}^{n}$ is a basis of $\text{Ker} (L)$. Let $G$ be the generator of $\mathcal C$
with the $i$-th row being $\mathbf c^i$ and $M=GG^T$. Let $G_1$ be the generator of $\mathcal C_1:=\text{Span}\{\mathbf c^i: i=h+1,h+2, \cdots, k\}$
with the $i$-th row being $\mathbf c^{h+i}$ and $M_{k-h,k-h}=G_1G_1^T$. Then, $M$ has the form
\begin{align*}
M=\left[  \begin{matrix}
  0_{h,h} & 0_{h,k-h}\\
  0_{k-h,h} & M_{k-h,k-h}
  \end{matrix}
 \right],
\end{align*}
where $0_{*,*}$ is the matrix with all entries being $0$.

Let $G_2$ be the matrix with the $i$-th row being $( l_1(\mathbf c^i), \cdots, l_t(\mathbf c^i))$ for $i\in \{1,\cdots, h\}$ and $M_2=G_2 G_2^{T}$.
Then $G_2$ is a generator matrix of $\mathcal C_L$. Then,
\begin{align*}
M=& \left[  \begin{matrix}
  0_{h,h}+ M_2 & 0_{h,k-h}\\
  0_{k-h,h} & M_{k-h,k-h}
  \end{matrix}
 \right]
 =\left[  \begin{matrix}
   M_2 & 0_{h,k-h}\\
  0_{k-h,h} & M_{k-h,k-h}
  \end{matrix}
 \right].
\end{align*}
Since $t<h$, $\det (M_2)=0$ from Lemma \ref{rank-GG}. Hence, $\det (M)=\det (M_2) \det (M_{k-h,k-h})=0$ and $\mathcal C_L$ is not an Euclidean LCD code.
\end{proof}

\begin{theorem}
Let $\mathcal C$ be an $[n,k]$-linear over $\mathbb F_q$ with $h:=h_{E}(\mathcal C)>0$
and $L$ be a linear transform from $\mathcal C$ to $\mathbb F_q^h$ such that
$\mathcal C= \text{Ker} (L) + \text{Hull}(\mathcal C)$. If the linear transform $L$   is surjective, then $\mathcal C_L$ (defined by Equation (\ref{C-L})) is an Euclidean LCD code.
\end{theorem}

\begin{corollary}
If there is an $[n,k,d]$-linear code $\mathcal C$ over $\mathbb F_q$ with $h:=h_{E}(\mathcal C)>0$, then there is a linear Euclidean LCD code with parameters $[n+h, k, \ge d]$.
\end{corollary}

\section{Concluding Remarks}\label{sect:concluding}
LCD codes have applications in information protection.   This paper is devoted to determine all possible Euclidean and Hermitian LCD  codes.
We completely determine all Euclidean LCD codes over $\mathbb F_q (q>3)$ and all Hermitian LCD codes over $\mathbb F_{q^2} (q>2)$   for all possible parameters.
More precisely,  we introduce a general construction of LCD codes from any linear codes. Further, we show that any linear code over
$\mathbb F_{q} (q>3)$ is equivalent to an Euclidean LCD code and any linear code over
$\mathbb F_{q^2} (q>2)$ is equivalent to a Hermitian LCD code. Consequently an $[n,k,d]$-linear  Euclidean LCD code over $\mathbb F_q$ with $q>3$ exists if there is an $[n,k,d]$-linear code over $\mathbb F_q$ and an $[n,k,d]$-linear Hermitian LCD code over $\mathbb F_{q^2}$ with $q>2$ exists if there is an $[n,k,d]$-linear code over $\mathbb F_{q^2}$. Hence, when $q>3$, $q$-ary Euclidean
LCD codes are as good as $q$-ary linear codes. When $q>2$, $q^2$-ary Hermitian
LCD codes are as good as $q^2$-ary linear codes. We finally, present a construction of LCD codes by extending linear codes. Our future work concerns the constructions and  the classification of Euclidean LCD codes over $\mathbb F_2$ or $\mathbb F_3$ and Hermitian LCD codes over $\mathbb F_4$.


\begin{thebibliography}{99}


\bibitem{BJ16}  K. Boonniyoma and S. Jitman.: Complementary dual subfield linear codes over finite fields, arXiv:1605.06827 [cs.IT], 2016.

\bibitem{CG14} C. Carlet and S. Guilley.: Complementary dual codes for counter-measures to side-channel attacks, In: E. R. Pinto et al. (eds.), Coding Theory and
Applications, CIM Series in Mathematical Sciences, vol. 3, pp. 97-105, Springer Verlag, 2014 and
Journal Adv. in Math. of Comm. 10(1), pp. 131-150, 2016.

\bibitem{CMTQ17} C. Carlet, S. Mesnager, C. Tang and Y. Qi.: Euclidean and Hermitian LCD MDS Codes, arXiv preprint arXiv:1702.08033, 2017.


\bibitem{DLL16} C. Ding, C. Li, and S. Li.: LCD Cyclic codes over finite fields, available at arXiv:1608. 0217v1 [cs.IT].

\bibitem{DNS17} H-Q. Dinh, B-T Nguyen, S. Sriboonchitta.: Constacyclic codes over finite commutative semi-simple rings, Journal Finite Fields and Their Applications,  Vol. 45, pp. 1-18, 2017.


\bibitem{DKO15} S.T. Dougherty, J.-L. Kim, B. Ozkaya, L. Sok and P. Sol\'e.: The combinatorics of LCD codes: Linear Programming bound and orthogonal matrices. To appear in Journal International Journal of Information and Coding Theory (IJICOT).

\bibitem{GKL17} L. Galvez, J-L Kim, N. Lee, Y-G. Roe, B-S Won.: Some Bounds on Binary LCD Codes,  arXiv preprint arXiv:1701.04165, 2017.


\bibitem{OOS17} C. G\"uneri, F. \"Ozbudak, B. \"Ozkaya, E. Sacikara, Z. Sepasdar and  P. Sol\'e.: Structure and performance of generalized quasi-cyclic codes, arXiv preprint arXiv:1702.00153, 2017.

\bibitem{GOS16} C. G\"uneri, B. \"Ozkaya, Sol\'e, Quasi-cyclic complementary dual codes. Journal Finite Fields and Their Applications, Vol. 42, pp.  67-80, 2016.


\bibitem{Jin16} L. Jin.: Construction of MDS codes with complementary duals,  IEEE Transactions on Information Theory, 2016.

\bibitem{JX17} L. Jin and C.P. Xing. Algebraic Geometry Codes with Complementary Duals  Exceed the Asymptotic Gilbert-Varshamov bound, arXiv preprint arXiv:1703.01441, 2017.

\bibitem{KSS12} W.V. Kandasamy, F. Smarandache, R. Sujatha, R. R. Duray.: Erasure Techniques in MRD codes. Infinite Study, 2012.

\bibitem{Li17} C. Li.: On Hermitian LCD codes from cyclic codes and their applications to orthogonal direct sum masking, arXiv preprint arXiv:1701.03986, 2017.


\bibitem{LDL16_0} S. Li, C. Ding, and H. Liu.: A family of reversible BCH codes, arXiv:1608.02169v1 [cs.IT].

\bibitem{LDL16_1} S. Li, C. Ding, and H. Liu.: Parameters of two classes of LCD BCH codes, arXiv:1608.02670 [cs.IT].

\bibitem{LL16} X. Liu X and H.  Liu.: Matrix-Product Complementary dual Codes, arXiv preprint arXiv:1604.03774, 2016.

\bibitem{mas92}  J. L. Massey.: Linear codes with complementary duals, Discrete Math., vol. 106-107, pp. 337-342, 1992.

\bibitem{MS77} F.J. MacWilliams and N. J. A. Sloane.: The theory of error-correcting codes. Elsevier, 1977.
\bibitem{MTQ16} S. Mesnager, C. Tang and Y. Qi.: Complementary dual algebraic geometry codes, arXiv preprint arXiv:1609.05649, 2016.




\bibitem{Sen04} N. Sendrier.: Linear codes with complementary duals meet the Gilbert-Varshamov bound. Discrete mathematics. 285 (1), pp. 345-347, 2004.

\bibitem{TH70} K. Tzeng and C. Hartmann.: On the minimum distance of certain reversible cyclic codes, IEEE Transactions on Information Theory, 16(5), pp. 644-646, 1970.


\bibitem{TVZ82} M.A.Tsfasman, S.G.Vl$\check{a}$dut, and T.Zink, Modular curves, Shimura curves, and Goppa codes, better than Varshamov-Gilbert bound, Math. Nachr., vol. 109, pp. 21-28, 1982.

\bibitem{YM94} X. Yang and J. L. Massey.: The condition for a cyclic code to have a complementary duala Journal Discrete Math., vol.
126, pp. 391-393, 1994.





\end{thebibliography}
\end{document}